\documentclass[final,12pt]{elsarticle}

\usepackage{amsmath,amssymb,amsthm}
\usepackage[a4paper]{geometry}
\usepackage[utf8]{inputenc}
\usepackage[T1]{fontenc}
\usepackage{lmodern}
\usepackage{mathrsfs}
\usepackage{microtype}
\usepackage{latexsym}
\usepackage[colorlinks]{hyperref}
\usepackage{cleveref}
\usepackage{graphicx}
\usepackage{bookmark}
\usepackage{booktabs}
\usepackage{enumitem}
\usepackage{diagbox}
\usepackage{lineno}
\usepackage{float}
\usepackage{color}
\usepackage{bm}
\usepackage{longtable}

\newtheorem{thm}{Theorem}
\newtheorem{lem}{Lemma}
\newtheorem{prop}{Proposition}

\newtheorem{coro}{Corollary}
\theoremstyle{definition}
\newtheorem{defi}{Definition}
\newtheorem{ex}{Example}
\newtheorem{rem}{Remark}

\DeclareMathOperator{\ord}{\mathrm{ord}}
\DeclareMathOperator{\wt}{\mathrm{wt}}

\DeclareMathOperator{\supp}{\mathrm{supp}}
\renewcommand{\vec}[1]{\boldsymbol{#1}}
\renewcommand{\d}{\mathrm{d}}

\begin{document}
	
\begin{frontmatter}



\title{Cyclic codes and cyclically covering subspaces \tnoteref{fund}
}
\tnotetext[fund]{The work of Minjia Shi was supported by the National Natural Science Foundation of China under Grant 12471490. 
}


\author[ahu1]{Xuan Wang}
\author[ahu1]{Minjia Shi\corref{CorAuthor}}
\cortext[CorAuthor]{Corresponding author.	E-mail address: smjwcl.good@163.com (M. Shi).}
%\author[sog1,sog2]{Junmin An}
%\author[sog1,sog2]{Jon-Lark Kim}
%% Author name

\affiliation[ahu1]{organization={Key Laboratory of Intelligent Computing Signal Processing, Ministry of Education},%Department and Organization
            addressline={School of Mathematical Sciences, Anhui University},
            city={Hefei},
            postcode={230601},
            state={Anhui},
            country={China}}
%\affiliation[ahu2]{organization={Key Laboratory of Intelligent Computing Signal Processing, Ministry of Education, Anhui University},%Department and Organization
%            addressline={111, Jiulong},
%            city={Hefei},
%            postcode={230601},
%            state={Anhui},
%            country={China}}
%\affiliation[sog1]{organization={Department of Mathematics, Sogang University},%Department and Organization
%            addressline={35, Baekbeom-ro},
%            city={Mapo-gu},
%            postcode={04107},
%            state={Seoul},
%            country={Republic of Korea}}
%\affiliation[sog2]{organization={Institute for Mathematical and Data Sciences, Sogang University},%Department and Organization
%            addressline={35, Baekbeom-ro},
%            city={Mapo-gu},
%            postcode={04107},
%            state={Seoul},
%            country={Republic of Korea}}

%% Abstract
\begin{abstract}
A subspace of $\mathbb{F}_q^n$ is called cyclically covering if the union of $\sigma^i(U)$ can cover the whole space $\mathbb{F}_q^n$, where $\sigma$ is the cyclic shift, $0 \leqslant i \leqslant n-1$. 
Let $h_q(n)$ be the largest possible co-dimension of a cyclically covering subspace of $\mathbb{F}_q^n$. 
We show that $h_2(2p) = 2$ for every prime $p$ such that $2$ is a primitive root modulo $p$. 
By constacyclic codes, we show that $h_q((q-1)n) = 0$ when $h_q(n) = 0$ and $\gcd(n,q-1) = 1$. 
We also derive a lower bound on $h_q(n)$ by the concept of support weight distribution, which is important in coding theory. 
Finally, using irreducible cyclic codes, we present several families of $n$ such that $h_q(n) = 0$. 
\end{abstract}



%% Keywords
\begin{keyword}
cyclically covering subspace \sep irreducible cyclic codes \sep constacyclic codes \sep support weight distribution \sep primitive root




\end{keyword}

\end{frontmatter}




\section{Introduction}

Let $\mathbb{F}_q$ be the finite field with $q$ elements where $q$ is a prime power. 
Let $\{ \vec{e}_0, \vec{e}_1, \dots, \vec{e}_{n-1} \}$ be the standard basis of the $n$-dimensional vector space $\mathbb{F}_q^n$, where $n$ is a positive integer. 
Throughout the paper, the indices of vectors in $\mathbb{F}_q^n$ will be taken modulo $n$. 
Define the \textit{cyclic shift} operator $\sigma: \mathbb{F}_q^n \to \mathbb{F}_q^n$ as 
\[ \sigma\left( \sum_{i=0}^{n-1} u_i \vec{e}_i \right) = \sum_{i=0}^{n-1} u_{i} \vec{e}_{i+1}, \]
where $\vec{u} = (u_0, u_1, \dots, u_{n-1}) \in \mathbb{F}_q^n$. 

A subspace $U$ of $\mathbb{F}_q^n$ is \textit{cyclically covering} if 
\[ \bigcup_{r=0}^{n-1} \sigma^r(U) = \mathbb{F}_q^n,  \]
where $\sigma^r$ is the composition of the map $\sigma$ with $r$ times. 
Note that $U' \subseteq \mathbb{F}_q^n$ is also cyclically covering if $U \subseteq U'$ is cyclically covering. 
A natural question is what is the smallest dimension of such a cyclically covering subspace? 
Or equivalently, to determine the value on 
\[ h_q(n) = n - \min \{ \dim U : U ~\text{is a cyclically covering subspace of}~ \mathbb{F}_q^n  \}. \]
Despite its simple definition, determining $h_q(n)$ is highly nontrivial.
Indeed, even for moderate values of $n$, many cases remain unknown (see \cite{JCTA-AGJ}). 
The main difficulty comes from the interplay between additive structure and cyclic symmetry, which makes standard linear-algebraic or combinatorial methods insufficient.
In 2019, Cameron, Ellis, and Raynaud \cite{EJC-CER} firstly proposed the function $h_q(n)$ and found that $h_q(n)$ is closely related to a well-known conjecture of Isbell \cite{MS-Isbell-1,PAMS-Isbell-2} and is a natural generalization of a problem posed by Cameron \cite{PBC-DM-C} in 1994.  

In fact, similar problems have been investigated before. 
For example, Luh \cite{AMASH-Luh} shows that every vector space (finite or infinite) over $\mathbb{F}_q$ can be expressed as a union of $q + 1$
proper subspaces in 1972. 
And in 1977, Jamison \cite{JCTA-J} determined, for each $0 < k < n$, the minimum number of $k$-flats (a $k$-flat is a translate of a $k$-dimensional subspace) that are required to cover $\mathbb{F}_q^n \setminus \{ \vec{0} \}$. 

An important open direction, posed by Cameron \cite[Problem 190]{PBC-DM-C}, is whether $h_2(n)$ remains bounded or grows unbounded along odd integers.
A breakthrough result by Aaronson, Groenland, and Johnston \cite{JCTA-AGJ} shows that if $p > 3$ is a prime such that $2$ is a primitive root modulo $p$, then $h_2(p) = 2$.
This connects the problem to Artin’s conjecture on primitive roots.
So, if there are infinitely many primes with $2$ as a primitive root, then $h_2(n) = 2$ for infinitely many odd $n$. 
Artin's conjecture is widely believed; in particular, it follows as a consequence of the generalised Riemann hypothesis, as shown by Hooley \cite{JRAM-H}. 

Usually, it is difficult to determine the exactly values on $h_q(n)$. 
Even for $n \leqslant 29$, there are several values on $h_2(n)$ are not determined \cite[Table 1]{JCTA-AGJ} and some values are obtained by brute force and computer search. 
In 2025, Sun, Ma, and Zeng \cite{FFTA-SMZ} developed new methods to determine $h_q(n)$ and found $h_2(n) = m-2$ when $3n = 2^m - 1 > 3$ and $m = \ord_n(2)$. 
In particular, they showed $h_2(21) = 4$, which is new according to \cite[Table 1]{JCTA-AGJ}.  

There are also several problems proposed in \cite[Section 7]{JCTA-AGJ}. 
One of them is: for which $n$, we have $h_q(n) = 0$ \cite{JCTA-AGJ,EJC-CER}? 
In 2024, Huang \cite{FFTA-H} firstly gave a necessary and sufficient condition by trace function under which $h_q(n) = 0$ when $\gcd(n,q) = 1$. 
Using the language of coding theory, this condition is equivalent to say ``$h_q(n) = 0$ if and only if every \textit{cyclic code} of length $n$ over $\mathbb{F}_q$ contains a codeword of Hamming weight $n$'',  
which has been proved via Discrete Fourier Transformers in \cite{arxiv-LY-2}. 
More details can also be seen in \cite{arxiv-LY-3}. 
This leads to deep connections between $h_q(n)$ and classical problems in coding theory, such as Hamming weight distribution and support weight distribution. 

An $[n,k]$ linear code over $\mathbb{F}_q$ is a $k$-dimensional subspace of $\mathbb{F}_q^n$. 
For each vector $\vec{v}$ in $\mathbb{F}_q^n$, its Hamming weight is defined as the number of nonzero coordinates, denoted by $\wt(\vec{v})$. 
Let $A_i$ be the number of codewords with Hamming weight $i$ in $\mathcal{C}$, where $0 \leqslant i \leqslant n$. The sequence $(A_0, A_1, \dots, A_n)$ is called the \textit{weight distribution} of $\mathcal{C}$ and the polynomial $A(z) = A_0 + A_1 z + \cdots A_{n-1} z^{n-1}$ is called the \textit{weight enumerator} of $\mathcal{C}$. 
A linear code $\mathcal{C}$ is called \textit{cyclic} if $\sigma(\mathcal{C}) = \mathcal{C}$. 
A fundamental and long-standing problem in coding theory is to determine its (Hamming) weight distribution. 
There are many related works, for example, irreducible cyclic codes \cite{TIT-D,DM-HKM,FFTA-SB,FFTA-SBR}, repeated-root cyclic codes \cite{TIT-L}, constacyclic codes \cite{FFTA-CFLL}, support weight distribution \cite{TIT-HKY,TIT-SLH}. 

The distribution of support weights, originally introduced by Helleseth and Kl\o ve \cite{DM-HKM,DM-Klove}, plays an important role in coding-theoretic applications such as the wire-tap channel of type II \cite{TIT-Wei}. 
However, in general, determining the support weight distribution of a code is a highly nontrivial problem, and only few explicit results are known.

Closely related is the study of lifted codes. Given an $[N,K]$ linear code $\mathcal{C}_1$ over $\mathbb{F}_q$, one may consider its lifted code $\mathcal{C}_\ell$ over $\mathbb{F}_{q^\ell}$ with the same generator matrix. 
A fundamental observation due to Helleseth et al. \cite{DM-HKM} (also independently by Greene \cite{SAM-G}) shows that the weight distribution of $\mathcal{C}_\ell$ can be expressed in terms of counting certain subspaces associated with the generator matrix.
Later, it was shown in \cite{TIT-HKY} that these subspace enumerations are equivalent to counting subcodes with prescribed support weights. 
Therefore, determining support weight distributions is equivalent to determine the Hamming weight distribution of a lifted code. 

In this paper, we focus on several classes of cyclic and constacyclic codes to study $h_q(n)$. 
Our main contributions are summarized as follows:

\begin{itemize}
	\item We determine $h_2(2p)=2$ when $2$ is a primitive root modulo an odd prime $p>3$.
	\item We show that $h_q(n(q-1))=0$ under natural coprimality conditions via constacyclic codes.
	\item We provide new method yielding $h_q((q^k-1)/(q-1)) = k-1$ under suitable arithmetic assumptions.
	\item We also characterize several families of $n$ for which $h_q(n)=0$.
\end{itemize}

This paper is organized as follows. 
\Cref{sec-pre} introduces preliminaries on cyclic codes and cyclically covering subspaces. 
In \Cref{sec-binary}, we prove the main result $h_2(2p) = 2$ when $2$ is a primitive root modulo the odd prime $p > 3$. 
Further applications by constacyclic codes and support weight distribution are given in \Cref{sec-two-app}. 
In \Cref{sec-ICC} presented several families of $n$ such that $h_q(n) = 0$ by irreducible cyclic codes. 
We conclude the paper in \Cref{sec-conclusion}.

\section{Preliminaries}\label{sec-pre}

\subsection{Cyclic codes}

Recall some basic definitions and notations in coding theory \cite{Huffman-Pless}. 

Let $\mathbb{F}_q^n$ be an $n$-dimensional vector space over the finite field $\mathbb{F}_q$, where $q$ is a prime power and $n$ is a positive integer. 
A $q$-ary $[n,k]$ \textit{linear code} $\mathcal{C}$ is a $k$-dimensional subspace of $\mathbb{F}_q^n$. 
The elements (vectors) in $\mathcal{C}$ are called \textit{codewords}. 
A linear code $\mathcal{C}$ is called \textit{cyclic} if for each codeword $\vec{u} = (u_0, u_1, \dots, u_{n-1}) \in \mathcal{C}$, we have $\sigma(\vec{u}) = (u_{n-1}, u_0, \dots, u_{n-2}) \in \mathcal{C}$. 
Consider the following map $\pi: \mathbb{F}_q^n \to \mathbb{F}_q[x] /  \langle (x^n - 1) \rangle$ defined as 
\[
(u_0, u_1, \dots, u_{n-1}) \mapsto u_0 + u_1 x + \dots + u_{n-1} x^{n-1}. 
\]
We can identify the codeword $\vec{u} \in \mathcal{C}$ with the polynomial $u_0 + u_1 x + \dots + u_{n-1} x^{n-1} \in \mathbb{F}_q[x] / \langle (x^n - 1) \rangle$.  
Then the linear code $\mathcal{C}$ is cyclic if and only if $\mathcal{C}$ is an ideal of $\mathbb{F}_q[x] / \langle (x^n - 1) \rangle$. 
The polynomial $g(x) \in \mathcal{C}$ is a \textit{generator polynomial} if $\langle g(x) \rangle = \mathcal{C}$. 
The cyclic code $\mathcal{C}$ is called \textit{irreducible} if $(x^n - 1)/g(x)$ is irreducible over $\mathbb{F}_q$. 

The \textit{support} of the vector $\vec{u} \in \mathbb{F}_q^n$ is defined as 
\[ \supp(\vec{u}) = \{ i: u_i \neq 0 \}. \]
The \textit{Hamming weight} of $\vec{u}$ is the size of its support, denoted by $\wt(\vec{u}) = |\supp(\vec{u})|$. 
The sequence $(A_0, A_1, \dots, A_n)$ of the code $\mathcal{C}$ is called the \textit{weight distribution} of $\mathcal{C}$, where 
\[ A_i = |\{ \vec{c} \in \mathcal{C}: \wt(\vec{c}) = i \}|. \]
The \textit{Hamming distance} of the vectors $\vec{u}$ and $\vec{v}$ is $\d(\vec{u}, \vec{v}) = \wt(\vec{u} - \vec{v})$. 
And the \textit{minimum distance} of the code $\mathcal{C}$ is defined as 
\[ \d(\mathcal{C}) = \min \{ \d(\vec{u}, \vec{v}): \vec{u}, \vec{v} \in \mathcal{C}, \vec{u} \neq \vec{v} \}. \]

\subsection{Cyclically covering subspaces}

Let $\sigma$ be the cyclic shift from $\mathbb{F}_q^n$ to $\mathbb{F}_q^n$ defined as above, and $\sigma^i$ be the map of $\sigma$ repeated $i$ times, where $0 \leqslant i \leqslant n-1$. 
A subspace $U$ of $\mathbb{F}_q^n$ is called \textit{cyclically covering} if 
\[ \bigcup_{i=0}^{n-1} \sigma^i(U) = \mathbb{F}_q^n. \]
For each positive integer $n$, let $h_q(n)$ be the largest possible co-dimension of a cyclically covering subspace of $\mathbb{F}_q^n$. 

In general, it is difficult to determine the values of $h_q(n)$. 
According to \cite{JCTA-AGJ,EJC-CER}, we have the following results. 
\begin{thm} \label{thm-known-1}
	Let $q$ be a power of the prime $p$, and $m,n$ be positive integers. 
	\begin{enumerate}[label=$\mathrm{(\arabic*)}$]
		\item $h_q(mn) \geqslant h_q(m) + h_q(n)$, 
		\item $h_q(pn) \leqslant p h_q(n)$, 
		\item $h_q(n) \leqslant \lfloor \log_q(n) \rfloor$, 
		\item $h_2(\ell) = 2$ if $\ell > 3$ is a prime for which $2$ as a primitive root. 
	\end{enumerate}
\end{thm}

As for $h_q(n) = 0$, we have 
\begin{thm}[\texorpdfstring{\cite{arxiv-LY-1}}{PDF string}]
	Let $q$ be a power of the prime $p$, and $n$ be a positive integer, $\gcd(n,q) = 1$.  
	Then $h_q(n) = 0$ if and only if $h_q(pn) = 0$. 
\end{thm}

\begin{rem}
	In fact, by \Cref{thm-known-1}, we have 
	\[ h_q(n) \leqslant h_q(pn) \leqslant p h_q(n). \]
	Thus, we can also obtain $h_q(n) = 0$ if and only if $h_q(pn) = 0$. 
	Here, we do not require $\gcd(n,q) = 1$. 
\end{rem}

From now on, we assume that $n$ and $q$ are coprime when determine $n$ such that $h_q(n) = 0$. 

\begin{thm}
	Let $\gcd(n,p) = 1$ and $q$ be a power of the prime $p$. 
	Then $h_q(n) = 0$ if 
	\begin{enumerate}[label=$\mathrm{(\arabic*)}$]
		\item $n \leqslant q$ (\cite[Corollary 6.4]{JCTA-AGJ})
		\item $n \mid (q-1)$  (\cite[Theorem 14]{EJC-CER})
		\item $n = q+1$ when $q$ is odd  (\cite[Corollary 3.8]{FFTA-H})
		\item $n = 2(q-1)$ when $4 \mid (q+1)$  (\cite[Corollary 3.9]{FFTA-H})
		\item $n > q$ is a prime with the prime $q$ as a primitive root  (\cite[Theorem 6.5]{JCTA-AGJ})
		\item $n = 2 \ell$ when $\gcd(q,\ell-1) = 1$ and $\ell$ is a prime such that $q$ is a primitive root modulo $2\ell$  (\cite[Corollary 3.11]{FFTA-H})
	\end{enumerate}
\end{thm}

\subsection{Relations between cyclic codes and cyclically covering subspaces}

Recall the concept ``work together'' defined in \cite{JCTA-AGJ}. 

\begin{defi} \label{def-ccs-vector}
	The vector $\vec{v}$ \textit{works} if for every vector $\vec{u} \in \mathbb{F}_{q}^{n}$, there exists a $k \in \{ 0, 1, \dots, n-1 \}$ such that 
	\[ \vec{v} \cdot \sigma^k \vec{u} = 0. \]
	Let $\vec{v}^{(1)}$, $\vec{v}^{(2)}$, $\cdots$, $\vec{v}^{(m)}$ be $m$ linearly independent vectors in $\mathbb{F}_{q}^{n}$. 
	Then these vectors \textit{work together} if for every vector $\vec{u} \in \mathbb{F}_{q}^{n}$, there exists a $k \in \{ 0, 1, \dots, n-1 \}$ such that 
	\[ \vec{v}^{(1)} \cdot \sigma^k \vec{u} = \vec{v}^{(2)} \cdot \sigma^k \vec{u} = \cdots = \vec{v}^{(m)} \cdot \sigma^k \vec{u} = 0. \]
\end{defi}

\begin{defi}\label{def-ccs-poly}
	The polynomial $\vec{f}$ \textit{works} if for every polynomial $\vec{g} \in \mathbb{F}_{q}[x] / \langle x^n - 1 \rangle$, there exists a $k \in \{ 0, 1, \dots, n-1 \}$ such that the coefficient of $x^k$ in the product $\vec{f} \cdot \vec{g}$ is $0$. 
	The polynomials $\vec{f}^{(1)}$, $\vec{f}^{(2)}$, $\cdots$, $\vec{f}^{(m)}$ in $\mathbb{F}_{q}[x] / \langle x^n - 1 \rangle$ \textit{work together} if for every polynomial $\vec{g} \in \mathbb{F}_{q}[x] / \langle x^n - 1 \rangle$, there exists a $k \in \{ 0, 1, \dots, n-1 \}$ such that the coefficients of $x^k$ in the products
	\[ \vec{f}^{(1)} \cdot \vec{g}, \vec{f}^{(2)} \cdot \vec{g}, \cdots, \vec{f}^{(m)} \cdot \vec{g} \]
	are all $0$. 
\end{defi}

If $(\vec{f}^{(i)} \cdot \vec{g})$ is regarded as a vector of length $n$, then we will use $(\vec{f}^{(i)} \cdot \vec{g})_j$ to denote the $j$-th coordinate, $0 \leqslant j \leqslant n-1$. 

\begin{rem}
	\Cref{def-ccs-vector} and \Cref{def-ccs-poly} are equivalent. 
	Moreover, for each $n$, we have $h_q(n) = m$ if and only if there exist at most $m$ linearly independent vectors $\vec{v}^{(1)}$, $\vec{v}^{(2)}$, $\cdots$, $\vec{v}^{(m)}$ in $\mathbb{F}_{q}^{n}$ that work together. 
\end{rem}

By the proof of \cite[Theorem 6.5]{JCTA-AGJ}, the vector $\vec{v}$ over $\mathbb{F}_q$ fails to work if and only if the corresponding polynomial $\vec{v}(x)$ is a factor (when modulo $x^n - 1$) of a polynomial with no zero coefficients. 
Then the ideal $I = \langle \vec{v}(x) \rangle$ (in fact, a cyclic code) of $\mathbb{F}_{q}^{n} / \langle x^n - 1 \rangle$ contains a vector of Hamming weight $n$. 
Therefore, we have the following necessary and sufficient condition such that $h_q(n) = 0$, which has also been proved via Discrete Fourier Transforms in \cite{arxiv-LY-2}. 

\begin{thm}
	For each $n$, we have $h_q(n) = 0$ if and only if every cyclic code of length $n$ over $\mathbb{F}_q$ contains a codeword of Hamming weight $n$. 
\end{thm}

\section{\texorpdfstring{$h_2(2p) = 2$}{}}\label{sec-binary}

Let $p \geqslant 5$ be a prime with $2$ as a primitive root. 
It is known that $h_2(p) = 2$ by \cite[Theorem 1.2]{JCTA-AGJ}, and then by \cite[Theorem 2.1 \& Theorem 2.3]{JCTA-AGJ}, we have 
\[ 2 = h_2(2) + h_2(p) \leqslant h_2(2p) \leqslant 2 h_2(p) = 4. \]
In this section, we will show 

\begin{thm}\label{thm-2p}
	If $p \geqslant 5$ is an odd prime with $2$ as a primitive root, then $h_2(2p) = 2$. 
\end{thm}

Assume that $x^{p} - 1 = (x-1) Q_p(x)$, where $Q_p(x) = 1 + x + x^2 + \dots + x^{p-1}$ is irreducible over $\mathbb{F}_2$ (see \cite{Book-Finite_Fields} for more details). 
There are only $2$ non-trivial cyclic codes of length $p$ over $\mathbb{F}_2$, i.e., 
\[ \mathcal{C}_1 = \langle (x-1) \rangle, \quad \mathcal{C}_p = \langle Q_p(x) \rangle. \]
Moreover, the vector $\vec{v} \in \mathbb{F}_{2}^{p}$ works if and only if $\vec{v} \in \mathcal{C}_1$. 
Next, consider the following factorization: 
\[ x^{2p} - 1 = (x-1)^2 (Q_p(x))^2. \]
The binary cyclic codes of length $n = 2p$ are precisely the ideals  
\[ C = \langle g(x) \rangle \subseteq \mathbb{F}_2[x] / (x^{2p} - 1), \]
where  
$g(x) \mid (x-1)^2 Q_p(x)^2$.  
Write \(g(x) = (x-1)^{e_0} Q_p(x)^{e_1}\) with  
\(0 \le e_0, e_1 \le 2\).  
Let $\mathcal{C}_{e_0, e_1}$ be the cyclic code generated by $g(x)$. 
Then by \cite[Remark, pp.10]{JCTA-AGJ}, the polynomial $\vec{f} \in \mathbb{F}_2^{2p}[x]/(x^{2p}-1)$ works if and only if $(x^2 + 1) \mid \vec{f}$. 
That is to say, the vector $\vec{f} \in \mathbb{F}_{2}^{2p}$ works if and only if $\vec{f} \in \mathcal{C}_{2,0}$ or $\mathcal{C}_{2,1}$, i.e., the polynomial $\vec{f}$ is divisible by $(x-1)^2$ or $(x-1)^2 Q_p(x) = (x-1)(x^p - 1)$. 

%All nine codes, together with their basic parameters and
%whether they contain a codeword of full support (Hamming weight \(2p\)), are listed below.
%\begin{table}[htbp]
%	\centering
%	\caption{Nine binary cyclic codes of length \(2p\)}
%	\label{tab:9codes}
%	\begin{tabular}{cccc}
%		\toprule
%		Code & Generator \(g(x)\) & Dimension & Contains a word of weight \(2p\)? \\
%		\midrule
%		\(\mathcal{C}_{0,0}\) & \(1\) & \(2p\) & Yes \\
%		\(\mathcal{C}_{1,0}\) & \(x-1\) & \(2p-1\) & Yes \\
%		\(\mathcal{C}_{2,0}\) & \((x-1)^2\) & \(2p-2\) & No \\
%		\(\mathcal{C}_{0,1}\) & \(Q_p(x)\) & \(p+1\) & Yes \\
%		\(\mathcal{C}_{1,1}\) & \((x-1)Q_p(x) = x^p - 1\) & \(p\) & Yes \\
%		\(\mathcal{C}_{2,1}\) & \((x-1)^2 Q_p(x)\) & \(p-1\) & No \\
%		\(\mathcal{C}_{0,2}\) & \(Q_p(x)^2\) & \(2\) & Yes \\
%		\(\mathcal{C}_{1,2}\) & \((x-1)Q_p(x)^2\) & \(1\) & Yes \\
%		\(\mathcal{C}_{2,2}\) & \(x^{2p} - 1\) (zero ideal) & \(0\) & No \\
%		\bottomrule
%	\end{tabular}
%\end{table}
%
%By \cite[Theorem 1]{TIT-L}, we present the codes of length $2p$ in Table \ref{tab:9codes}. 
%So, the vector $\vec{v} \in \mathbb{F}_{2}^{2p}$ works if and only if $\vec{v} \in \mathcal{C}_{2,0}$ or $\mathcal{C}_{2,1}$, i.e., the polynomial $\vec{v}(x)$ is divisible by $(x-1)^2$ or $(x-1)^2 Q_p(x) = (x-1)(x^p - 1)$. 

One can easily check that $\mathcal{C}_{2,1} = \{ (\vec{c}, \vec{c}) : \vec{c} \in \mathcal{C}_1 \}$, and $\mathcal{C}_{2,1} \subseteq \mathcal{C}_{2,0}$. 

To prove \Cref{thm-2p}, we need several lemmas. 

\begin{lem}
	For each $\vec{f} \in \mathcal{C}_{2,1} \setminus \{ \vec{0} \}$ and each $\vec{g} \in \mathcal{C}_{2,0} \setminus \mathcal{C}_{2,1}$, we have $\langle \vec{f} \rangle = \mathcal{C}_{2,1}$ and $\langle \vec{g} \rangle = \mathcal{C}_{2,0}$. 
\end{lem}
\begin{proof}
	Since $\mathcal{C}_{2,1}$ has no proper ideals, then $\langle \vec{f} \rangle = \mathcal{C}_{2,1}$. 
	As for $\mathcal{C}_{2,0}$, it only has one proper ideal $\mathcal{C}_{2,1}$. 
	So $\langle \vec{g} \rangle = \mathcal{C}_{2,0}$ since $\vec{g} \in \mathcal{C}_{2,0} \setminus \mathcal{C}_{2,1}$. 
\end{proof}

\begin{coro}
	For each $\vec{f}_1, \vec{f}_2 \in \mathcal{C}_{2,1} \setminus \{ \vec{0} \}$ and each $\vec{g}_1, \vec{g}_2 \in \mathcal{C}_{2,0} \setminus \mathcal{C}_{2,1}$, there exist polynomials $\vec{\lambda} \in \mathcal{C}_{2,1} \setminus \{ \vec{0} \}$ and $\vec{\mu} \in \mathcal{C}_{2,0} \setminus \mathcal{C}_{2,1}$ such that $\vec{f}_2 = \vec{f}_1 \vec{\lambda}$ and $\vec{g}_2 = \vec{g}_1 \vec{\mu}$. 
\end{coro}

\begin{rem}
	Note that for each nonzero $\vec{g} \in \mathcal{C}_{2,0}$, we have 
	\[ \vec{g} \cdot (Q_p^2(x) + 1) = \vec{g}. \]
	Since 
	\[ (x^2 + 1) \mid (Q_p^2(x) + 1) = x^2 + x^4 + \cdots + x^{2p-2}, \]
	then $\vec{1}_{\mathcal{C}_{2,0}} = Q_p^2(x) + 1 \in \mathcal{C}_{2,0} \setminus \mathcal{C}_{2,1}$ can be regarded as the ``identity element'' of $\mathcal{C}_{2,0}$. 
	Moreover, for each $\vec{g} \in \mathcal{C}_{2,0} \setminus \mathcal{C}_{2,1}$, there exists a $\vec{g}'$ such that $\vec{g} \vec{g}' = \vec{1}_{\mathcal{C}_{2,0}}$. 
%	So we denote $\vec{g}'$ by the inverse $\vec{g}^{-1}$. 
\end{rem}

%\begin{rem}
%	Similarly, for each nonzero $\vec{f} \in \mathcal{C}_{2,1}$, we have 
%	\[ \vec{f} \cdot (Q_p + 1) = \vec{f}. \]
%	Then $\vec{1}_{\mathcal{C}_{2,1}} = Q_p(x) + 1 \notin \mathcal{C}_{2,1}$. 
%	Moreover, for each nonzero $\vec{f} \in \mathcal{C}_{2,1}$, there exists a $\vec{f}' \in \mathcal{C}_{2,1}$ such that $\vec{f} \vec{f}' = \vec{1}_{\mathcal{C}_{2,1}}$. 
%	So we denote $\vec{f}'$ by the inverse $\vec{f}^{-1}$. 
%\end{rem}

\begin{lem}
	Let $\vec{f}, \vec{u} \in \mathbb{F}_2[x]/\langle (x^n - 1) \rangle$, where $\vec{f}, \vec{u} \neq \vec{0}$. 
	Assume that $\vec{f}$ is the generator of the cyclic code $\mathcal{C}$. 
	Then the set \[ S = \{ \vec{s}\in \mathbb{F}_{2}[x]/\langle (x^n - 1) \rangle: \vec{f} \vec{s} = \vec{u} \} \]
	is an affine space with dimension $\dim(S) = \dim(V)$, where
	\[ V = \{ \vec{s}\in \mathbb{F}_{2}[x]/\langle (x^n - 1) \rangle: \vec{f} \vec{s} = \vec{0} \} \] 
	is a vector space over $\mathbb{F}_2$. 
\end{lem}
\begin{proof}
	It is easy to check that $V$ is a vector space over $\mathbb{F}_2$. 
	If $S$ is regarded as a coset of $\mathbb{F}_2^n / V$, then it is an affine space. 
\end{proof}

\begin{rem}
	The set $V_i = \{ \vec{s}\in \mathbb{F}_{2}[x]/\langle (x^n - 1) \rangle: (\vec{f} \vec{s})_i = \vec{0} \}$ is also a vector space over $\mathbb{F}_2$, where $(\vec{f} \vec{s})_i$ is the $i$-th coordinate of $\vec{f} \vec{s}$. 
	Moreover, the dimension of $V_i$ is exactly $n-1$ since 
	\[ (\vec{f} \vec{s})_i = \sum_{j=0}^{n-1} f_j s_{i-j} = 0 \]
	is a hyperplane, where $\vec{f} = \sum_{j=0}^{n-1} f_j x^j$ and $\vec{s} = \sum_{j=0}^{n-1} s_j x^j$. 
\end{rem}

By linear algebra, we have the following result and we omit the proof here. 

\begin{lem}
	Let $U$ and $V$ be two linear subspaces of $\mathbb{F}_2^n$, respectively. 
	Let $A = \vec{v} + V$ for some $\vec{v} \in \mathbb{F}_2^n \setminus V$. 
	If $A \cap U$ is non-empty, then $A \cap U$ is an affine space with dimension
	\[ \dim(U) + \dim(V) - \dim(U+V), \]
	where $\dim(U+V)$ is the dimension of the vector space spanned by $U$ and $V$. 
\end{lem}

\begin{prop}\label{prop-type-21}
	Given two different nonzero polynomials $\vec{f}_1, \vec{f}_2 \in \mathcal{C}_{2,1}$, $\vec{g} \in \mathcal{C}_{2,0} \setminus \mathcal{C}_{2,1}$, then there exists a polynomial $\vec{h}$ such that 
	\begin{equation} \label{eq-binary-1}
		\supp(\vec{f}_1 \vec{h}) \cup \supp(\vec{f}_2 \vec{h}) \cup \supp(\vec{g} \vec{h}) = \mathbb{Z}_{2p}. 
	\end{equation}
\end{prop}
\begin{proof}
	It is known that there exists a polynomial $\vec{s}_0$ such that $\vec{f}_1 \vec{s}_0 = (x^p - 1) Q_p(x)$, whose support is exactly $\mathbb{Z}_{2p} \setminus \{ 0, p\}$. 
	So we only need to consider the existence of $\vec{s}_0$ such that 
	\begin{equation} \label{eq-binary-1-1}
		\{ 0, p \} \subseteq \supp(\vec{f}_2 \vec{s}_0) \cup \supp(\vec{g} \vec{s}_0). 
	\end{equation} 
	\begin{enumerate}[label=$\mathrm{(\arabic*)}$]
		\item 
		Let $\vec{u}$ be the vector corresponding to $(x^p - 1) Q_p(x)$. 
		Consider the following solution space 
		\[ S = \{ \vec{s} \in \mathbb{F}_{2}^{2p}: \vec{f}_1 \vec{s} = \vec{u} \}. \]
		Note that $\vec{f}_1 \vec{s}' = \vec{0}$ if and only if $\vec{s}'$ is divisible by $Q_p$, i.e., $\vec{s}' \in \mathcal{C}_{0,1}$. 
		Thus, $S$ is a affine space with dimension $2p - \deg(Q_p) = p+1$. 
		Moreover, Equation \eqref{eq-binary-1-1} is true if and only if 
		\begin{equation} \label{eq-binary-1-2}
			\left( (\vec{f}_2 \vec{s})_0, (\vec{g} \vec{s})_0 \right) \neq (0, 0), 
			\quad \text{and} \quad 
			\left( (\vec{f}_2 \vec{s})_p, (\vec{g} \vec{s})_p \right) \neq (0, 0),  
		\end{equation}
		where $(\vec{f}_2 \vec{s})_i$ and $(\vec{g} \vec{s})_i$ denote the $i$-th coordinate of $\vec{f}_2 \vec{s}$ and $\vec{g} \vec{s}$, respectively, $i \in \{ 0, p\}$. 
		
		\item 
		Define the following two sets with $i \in \{ 0, p\}$: 
		\[
		B_i = \{ \vec{s} \in S: (\vec{f}_2 \vec{s})_i = (\vec{g} \vec{s})_i = 0 \}. 
		\]
		So $B_i \subseteq S$. 
		Let 
		\[ V_i =  \{ \vec{s} \in \mathbb{F}_{2}^{2p}: (\vec{f}_2 \vec{s})_i = (\vec{g} \vec{s})_i = 0 \}, \]
		which is a $(2p-2)$-dimensional linear subspace of $\mathbb{F}_{2}^{2p}$, since $\vec{f}_2$ and $\vec{g}$ belong to different cyclic codes. 
		Moreover, $B_i = V_i \cap S$ is an affine space. 
		
		\item 
		We claim $B_i \neq S$, i.e., $B_i$ is a proper subset of $S$. 
		For otherwise, if $B_0 = S$, then for all $\vec{s} \in S$, we have $(\vec{g} \vec{s})_0 = 0$. 
		Take $\vec{s}_0 \in S$ such that $\vec{f}_1 \vec{s}_0 = \vec{u}$. 
		It is known that for all $\vec{s}' \in \mathcal{C}_{0,1}$ which is generated by $Q_p$, we have $\vec{f}_1 \vec{s}' = \vec{0}$. 
		Since $\vec{s} = \vec{s}_0 + Q_p \vec{t} \in S$ for every polynomial $\vec{t}$, then $(\vec{g} \vec{s})_0 = (\vec{g} Q_p \vec{t})_0 = 0$ due to $(\vec{g} \vec{s}_0)_0 = 0$, which implies that $\vec{g} Q_p = 0$, i.e., $\vec{g} \in \mathcal{C}_{2,1}$, a contraction. 
		Similarly, we have $B_p \neq S$.  
		Note that $B_0$ and $B_p$ are affine spaces, then 
		\[ |B_0| \leqslant 2^p, \quad |B_p| \leqslant 2^p. \]
		
		\item 
		We claim $B_0 \cup B_p \neq S$. 
		It is clear that $B_0 \cap B_p = (V_0 \cap V_p) \cap S$ is an affine space. 
		Since $\dim(V_0 \cap V_p) \geqslant 2p-4$ and $\dim(\langle V_0 \cap V_p, S \rangle) \leqslant 2p$, 
		then 
		\begin{align*}
			\dim(B_0 \cap B_p) 
			& = \dim(V_0 \cap V_p) + \dim(S) - \dim(\langle V_0 \cap V_p, S \rangle) \\
			& \geqslant (2 p - 4) + (p + 1) - 2p = p-3, 
		\end{align*}
		since $\dim(V_0 \cap V_p) = 2p-4$, $\dim(S) = p$, and $\dim(\langle V_0 \cap V_p, S \rangle) = \dim(\langle V_0 \cap V_p, \mathcal{C}_{0,1} \rangle) \leqslant 2p$. 
		Hence, when $p > 3$, we have 
		\[ |B_0 \cup B_p| = |B_0| + |B_p| - |B_0 \cap B_p| \leqslant 2^{p+1} - 2^{p-3} < 2^{p+1} = |S|, \]
		which implies that there exists a polynomial $\vec{s} \in S \setminus (B_0 \cup B_p)$ such that Equation \eqref{eq-binary-2-2} holds. 
	\end{enumerate}
	In a word, we gives the existence of $\vec{s}$ such that Equation \eqref{eq-binary-1-1} holds. 
\end{proof}

\begin{rem}
	Here, we require that $\vec{f}_1 \neq \vec{f}_2$. 
	For otherwise, Equation \eqref{eq-binary-1} will not hold. 
	For example, when $p=3$, we have $h_2(6) = 2$ and there exist two polynomials $\vec{f}$ and $\vec{g}$ that work together, where 
	\[ \vec{f} = x + x^2 + x^4 + x^5 \in \mathcal{C}_{2,1}, \quad \vec{g} = x^3 + x^5 \in \mathcal{C}_{2,0} \setminus \mathcal{C}_{2,1}.  \]
\end{rem}

\begin{prop}\label{prop-type-12}
	Given two different polynomials $\vec{g}_1, \vec{g}_2 \in \mathcal{C}_{2,0} \setminus \mathcal{C}_{2,1}$, and one nonzero polynomial $\vec{f} \in \mathcal{C}_{2,1}$, there exists a nonzero polynomial $\vec{h}$ such that 
	\begin{equation} \label{eq-binary-2}
		\supp(\vec{f} \vec{h}) \cup \supp(\vec{g}_1 \vec{h}) \cup \supp(\vec{g}_2 \vec{h}) = \mathbb{Z}_{2p}. 
	\end{equation}
\end{prop}
\begin{proof}
	It is known that there exists a polynomial $\vec{s}_0$ such that $\vec{f} \vec{s}_0 = (x^p - 1) Q_p(x)$, whose support is exactly $\mathbb{Z}_{2p} \setminus \{ 0, p\}$. 
	So we only need to consider the existence of $\vec{s}_0$ such that 
	\begin{equation} \label{eq-binary-2-1}
		\{ 0, p \} \subseteq \supp(\vec{g}_1 \vec{s}_0) \cup \supp(\vec{g}_2 \vec{s}_0). 
	\end{equation} 
	\begin{enumerate}[label=$\mathrm{(\arabic*)}$]
		\item 
		Let $\vec{u} = (\vec{a}, \vec{a})$ be the vector corresponding to $(x^p - 1) Q_p(x)$. 
		Consider the following solution space 
		\[ S = \{ \vec{s} \in \mathbb{F}_{2}^{2p}: \vec{f} \vec{s} = \vec{u} \}. \]
		Note that $\vec{f} \vec{s}' = \vec{0}$ if and only if $\vec{s}'$ is divisible by $Q_p$, i.e., $\vec{s}' \in \mathcal{C}_{0,1}$. 
		Thus, $S$ is a affine space with dimension $2p - \deg(Q_p) = p+1$. 
		Moreover, Equation \eqref{eq-binary-2-1} is true if and only if 
		\begin{equation} \label{eq-binary-2-2}
			\left( (\vec{g}_1 \vec{s})_0, (\vec{g}_2 \vec{s})_0 \right) \neq (0, 0), 
			\quad \text{and} \quad 
			\left( (\vec{g}_1 \vec{s})_p, (\vec{g}_2 \vec{s})_p \right) \neq (0, 0),  
		\end{equation}
		where $(\vec{g}_1 \vec{s})_i$ and $(\vec{g}_2 \vec{s})_i$ denote the $i$-th coordinate of $\vec{g}_1 \vec{s}$ and $\vec{g}_2 \vec{s}$, respectively, $i \in \{ 0, p\}$. 
		
		\item 
		Define the following two sets with $i \in \{ 0, p\}$: 
		\[
		B_i = \{ \vec{s} \in S: (\vec{g}_1 \vec{s})_i = (\vec{g}_2 \vec{s})_i = 0 \}. 
		\]
		So $B_i \subseteq S$. 
		Let 
		\[ V_i =  \{ \vec{s} \in \mathbb{F}_{2}^{2p}: (\vec{g}_1 \vec{s})_i = (\vec{g}_2 \vec{s})_i = 0 \}, \]
		which is a $(2p-2)$-dimensional linear subspace of $\mathbb{F}_{2}^{2p}$. 
		Moreover, $B_i = V_i \cap S$ is an affine space. 
		
		\item 
		We claim $B_i \neq S$, i.e., $B_i$ is a proper subset of $S$. 
		For otherwise, if $B_0 = S$, then for all $\vec{s} \in S$, we have $(\vec{g}_1 \vec{s})_0 = 0$. 
		Take $\vec{s}_0 \in S$ such that $\vec{f} \vec{s}_0 = \vec{u}$. 
		It is known that for all $\vec{s}' \in \mathcal{C}_{0,1}$ which is generated by $Q_p$, we have $\vec{f} \vec{s}' = \vec{0}$. 
		Since $\vec{s} = \vec{s}_0 + Q_p \vec{t} \in S$ for every polynomial $\vec{t}$, then $(\vec{g}_1 \vec{s})_0 = (\vec{g}_1 Q_p \vec{t})_0 = 0$ due to $(\vec{g}_1 \vec{s}_0)_0 = 0$, which implies that $\vec{g}_1 Q_p = 0$, i.e., $\vec{g}_1 \in \mathcal{C}_{2,1}$, a contraction. 
		Similarly, we have $B_p \neq S$.  
		Note that $B_0$ and $B_p$ are affine spaces, then 
		\[ |B_0| \leqslant 2^p, \quad |B_p| \leqslant 2^p. \]
		
		\item 
		We claim $B_0 \cup B_p \neq S$. 
		It is clear that $B_0 \cap B_p = (V_0 \cap V_p) \cap S$ is an affine space. 
		Since $\dim(V_0 \cap V_p) \geqslant 2p-4$ and $\dim(\langle V_0 \cap V_p, S \rangle) \leqslant 2p$, 
		then 
		\begin{align*}
			\dim(B_0 \cap B_p) & = \dim(V_0 \cap V_p) + \dim(S) - \dim(\langle V_0 \cap V_p, S \rangle) \\
			& \geqslant (2 p - 4) + (p + 1) - 2p = p-3. 
		\end{align*}
		Hence, when $p > 3$, we have 
		\[ |B_0 \cup B_p| = |B_0| + |B_p| - |B_0 \cap B_p| \leqslant 2^{p+1} - 2^{p-3} < 2^{p+1} = |S|, \]
		which implies that there exists a polynomial $\vec{s} \in S \setminus (B_0 \cup B_p)$ such that Equation \eqref{eq-binary-2-2} holds. 
	\end{enumerate}
	In a word, we gives the existence of $\vec{s}$ such that Equation \eqref{eq-binary-2-1} holds. 
\end{proof}

\begin{proof}[\textbf{Proof of Theorem \ref{thm-2p}}]
	If $h_2(2p) = m \geqslant 3$, then there exist $m$ linearly independent vectors $\vec{v}^{(1)}$, $\vec{v}^{(2)}$, $\cdots$, $\vec{v}^{(m)} \in \mathcal{C}_{2,0}$ that work together since $\mathcal{C}_{2,1} \subseteq \mathcal{C}_{2,0}$. 
	
	Let $\vec{v}^{(i)} = (\vec{u}^{(i)} + \vec{w}^{(i)}, \vec{w}^{(i)})$, where $\vec{u}^{(i)}, \vec{w}^{(i)} \in \mathbb{F}_{2}^{p}$. 
	Without loss of generality, let $\{ \vec{u}^{(1)}, \vec{u}^{(2)}, \dots, \vec{u}^{(\ell)} \}$ be the maximal linearly independent subset of $\{ \vec{u}^{(1)}, \vec{u}^{(2)}, \dots, \vec{u}^{(m)} \}$. 
	Then by the proof of \cite[Theorem 2.3]{JCTA-AGJ}, we have $h_2(p) = 2$ and the following facts: 
	\begin{itemize}
		\item $\ell \leqslant h_2(p) = 2$, $m - \ell \leqslant h_2(p) = 2$, $2 \leqslant m \leqslant 4$, 
		
		\item $\vec{v}^{(i)} = (\vec{w}^{(i)}, \vec{w}^{(i)})$ for $i > \ell$, 
		
		\item the vectors $\vec{w}^{(i)}$ work together for length $n$ and $i > \ell$. 
	\end{itemize}
	Thus, $\vec{v}^{(i)} \in \mathcal{C}_{2,1}$ for $i > \ell$ and $\vec{v}^{(i)} \notin \mathcal{C}_{2,1}$ for $i \leqslant \ell$ since $\vec{u}^{(i)} \neq \vec{0}$. 
	We have the following cases. 
	\begin{enumerate}[label=$\mathrm{(\arabic*)}$]
		\item 
		If $\ell = 1$, then $\vec{v}^{(1)} \in \mathcal{C}_{2,0} \setminus \mathcal{C}_{2,1}$ and $\vec{v}^{(2)}, \vec{v}^{(3)} \in \mathcal{C}_{2,1}$. 
		By \Cref{prop-type-21}, these $3$ vectors can not work together. 
		
		\item 
		If $\ell = 2$, then $\vec{v}^{(1)}, \vec{v}^{(2)} \in \mathcal{C}_{2,0} \setminus \mathcal{C}_{2,1}$ and $\vec{v}^{(3)} \in \mathcal{C}_{2,1}$. 
		By \Cref{prop-type-12}, these $3$ vectors can not work together. 
	\end{enumerate}
	Therefore, $h_2(2p) = 2$ when $2$ is a primitive root modulo $p$. 
\end{proof}

\section{Other applications of coding theory} \label{sec-two-app}

In this section, we will use constacyclic codes and support weight distribution to obtain more values on $h_q(n)$. 

\subsection{\texorpdfstring{$\lambda$}{}-constacyclic codes}

In this part, we will show $h_q(kn) = 0$ for several families of $k$ if $h_q(n) = 0$. 
\begin{thm} \label{thm-q-ary-q-1}
	Let $q$ be a prime power and $n$ be a positive integer such that $h_q(n) = 0$. 
	If $\gcd(n, q-1) = 1$, then $h_q((q-1)n) = 0$. 
\end{thm}

To prove the above result, we need the concept of $\lambda$-constacyclic codes, which is an ideal of $\mathbb{F}_q[x]/\langle (x^n - \lambda) \rangle$, where $\lambda \in \mathbb{F}_{q}^{\ast}$. 
See \cite{FFTA-CFLL} for more details. 

For $\lambda, \mu \in \mathbb{F}_{q}^{\ast}$, the $\mathbb{F}_q$-algebra isomorphism
\[ \Psi: \mathbb{F}_{q}[x]/\langle (x^n - \mu) \rangle \longrightarrow \mathbb{F}_{q}[x]/\langle (x^n - \lambda) \rangle \]
is called \textbf{isometry} if it preserves the Hamming distances on the algebras, i.e., 
\[ \d_H(\Psi(\vec{a}), \Psi(\vec{b})) = \d_H(\vec{a}, \vec{b}), \quad \forall \vec{a}, \vec{b} \in \mathbb{F}_{q}[x]/\langle (x^n - \mu) \rangle.  \]
Hence, all the $\lambda$-constacyclic codes of length $n$ are one-to-one corresponding to all the $\mu$-constacyclic codes of length $n$ such that the corresponding constacyclic codes have the same weight distribution. 

\begin{prop}[\texorpdfstring{\cite[Corollary 3.4 \& Corollary 3.5]{FFTA-CFLL}}{}]
	If there exists an element $a \in \mathbb{F}_{q}^{\ast}$ such that $a^n \lambda = 1$, then the $\lambda$-constacyclic codes are isometric to the cyclic codes of length $n$, where the isometry map $\Psi_a$ from $\mathbb{F}_{q}[x]/\langle (x^n - 1) \rangle$ to $\mathbb{F}_{q}[x]/\langle (x^n - \lambda) \rangle$ is defined as 
	\[ \Psi_a (f(x)) = f(ax) \]
	In particular, if $\gcd(n,q-1) = 1$, then for every $\lambda \in \mathbb{F}_{q}^{\ast}$, the $\lambda$-constacyclic codes are isometric to the cyclic codes of length $n$. 
\end{prop}

\begin{proof}[\textbf{Proof of Theorem \ref{thm-q-ary-q-1}}]
	Since $\gcd(n,q-1) = 1$, then for every $\lambda \in \mathbb{F}_{q}^{\ast}$, the $\lambda$-constacyclic codes are isometric to the cyclic codes of length $n$. 
	Considering that $h_q(n) = 0$, all the $\lambda$-constacyclic codes corresponding to the irreducible cyclic codes contain a codeword of Hamming weight $n$. 
	
	Let $g(x)$ be a irreducible factor $x^{(q-1)n} - 1$. 
	Since 
	\[ x^{(q-1)n} - 1 = (x^n)^{q-1} - 1 = \prod_{\alpha \in \mathbb{F}_q^{\ast}} (x^n - \alpha), \]
	then $g(x)$ divides $x^n - a$ for some $a \in \mathbb{F}_{q}^{\ast}$. 
	Thus the constacyclic code of length $n$ over $\mathbb{F}_q$ generated by $g(x)$ contains a codeword of Hamming weight $n$ due to $h_q(n) = 0$ and $\gcd(n,q-1) = 1$. 
	That is to say, there exists a polynomial $h$ such that $c(x) = h(x) \cdot \dfrac{x^n-a}{g(x)}$ has no zero coefficients, where $\deg(h) < n - \deg(g)$. 
	Hence, the irreducible cyclic code of length $(q-1)n$ over $\mathbb{F}_q$ generated by $(x^{(q-1)n} - 1) / g(x)$ contains a codeword
	\begin{align*}
		h \cdot \frac{x^{(q-1)n} - 1}{g(x)} & = h \cdot \frac{x^n-a}{g(x)} \cdot (y^{q-2} + a y^{q-3} + a^2 y^{q-4} + \cdots + a^{q-1} x + a^{q-2}) \\
		& = a^{q-2} c(x) + a^{q-1} x^n c(x) + \cdots + a x^{(q-3) n} c(x) + x^{(q-2) n} c(x)
	\end{align*}
	which has no zero coefficients, where $y = x^n$. 
	Therefore, $h_q((q-1)n) = 0$. 
\end{proof}

%Let $r = q^s$, $1 < N \mid (r-1)$ and $n = (r-1)/N$, where $q$ is a prime power. 
%Let $\alpha$ be a primitive element of $\mathbb{F}_r$ and let $\theta = \alpha^N$. 
%Then each irreducible cyclic code of length $n$ over $\mathbb{F}_q$ can be represented as follows
%\[ \mathcal{C}(r, N) = \{ \tr_{r/q}(\beta), \tr_{r/q}(\beta \theta), \dots, \tr_{r/q}(\beta \theta^{n-1}) : \beta \in \mathbb{F}_r \}  \]
%is an irreducible cyclic code \cite{TIT-D}, where $\tr_{r/q}$ is the trace function from $\mathbb{F}_r$ onto $\mathbb{F}_q$. 
%Its size is 
%Note that $\theta$ is the root of some irreducible polynomial $g(x) \mid (x^n - 1)$. 
%It is clear that $\mathcal{C}(r, N)$ contains a codeword of Hamming weight $n$ if and only if there exists a $\beta \in \mathbb{F}_r$ such that $\tr_{r/q}(\beta \theta^i) \neq 0$ for each $0 \leqslant i \leqslant n - 1$.

\subsection{support weight}

Let $\mathcal{C}$ be an $[n,k]$ linear code over $\mathbb{F}_q$, and $\mathcal{D}$ be a subcode of $\mathcal{C}$ with $\dim(\mathcal{D}) = r$. 
Let $\chi(\mathcal{D})$ be the \textbf{support weight} of $\mathcal{D}$ defined as follows
\[ \chi(\mathcal{D}) = |\{ i: \text{there exists } (c_1, c_2, \dots, c_n) \in \mathcal{D} \text{ such that } c_i \neq 0 \}|. \]
For $0 \leqslant r \leqslant k$ and $0 \leqslant i \leqslant n$, let $A_{i}^{r}$ be the number of $r$-dimensional subcodes of $\mathcal{C}$ of support weight $i$. 
The $r$-th support weight distribution is the sequence
\[ A_{0}^{r}, A_{1}^{r}, \dots, A_{n}^{r}. \]
%Denote the maximum (minimum) support weight by $D_r$ ($d_r$). 

\begin{prop}
	Let $\mathcal{C}$ be a cyclic code of length $n$ over $\mathbb{F}_q$. 
	If $A_{n}^{r} = 0$ for some positive integer $r$, then $h_q(n) \geqslant r$. 
\end{prop}
\begin{proof}
	Let $\vec{f}^{(1)}$, $\vec{f}^{(2)}$, $\cdots$, $\vec{f}^{(r)} \in \mathbb{F}_q^n$ be $m$ linearly independent vectors (polynomials).
	Since $A_{n}^{r} = 0$, then for each polynomial $\vec{g}$, there exists a $k \in \mathbb{Z}_n$ such that 
	\[ (\vec{f}^{(1)} \cdot \vec{g})_k = (\vec{f}^{(2)} \cdot \vec{g})_k = \cdots = (\vec{f}^{(r)} \cdot \vec{g})_k = 0, \]
	which implies that these vectors work together. 
	Hence, $h_q(n) \geqslant r$. 
\end{proof}

\begin{ex}
	For the $[n = \frac{q^k-1}{q-1}, k]$ Simplex code $\mathcal{S}_k$, according to \cite[Theorem 3]{TIT-SLH}, its support weight distribution $A_i^{j} \neq 0$ if and only if 
	\[ i = \frac{q^k - q^{k-j}}{q-1}. \]
	For all $j < k$, we have $A_i^j > 0$. 
	Thus, if $\mathcal{S}_k$ is a cyclic code, i.e., $\gcd(k, q-1) = 1$ (see \cite{Huffman-Pless} for more details), then $h_q(n) \geqslant k-1$. 
	Hence, $h_q(n) = k-1$ since $h_q(n) \leqslant \lfloor \log_q(n) \rfloor = k - 1$ by \cite[Lemma 4]{EJC-CER}. 
\end{ex}

\section{irreducible cyclic codes and \texorpdfstring{$h_q(n) = 0$}{}}\label{sec-ICC}

Let $\mathbb{F}_q$ be the finite field with $q$ elements and $n$ be a positive integer co-prime to $q$. 
Let $\ord_n(q)$ be the \textbf{multiplicative order} of $q$ modulo $n$, which is the smallest positive integer $k$ such that $q^k \equiv 1 \pmod{n}$. 
Let $\theta_n$ be a primitive $n$-th root of unity in some extension field of $\mathbb{F}_q$. 
The $q$-cyclotomic coset modulo $n$ containing $i$ is defined by 
\[ \mathcal{C}_i = \{ i \cdot q^j \pmod{n}: j = 0,1,\dots,n-1 \}. \]
Corrseponding to the $q$-cyclotomic coset $\mathcal{C}_i$, let 
\[ M_s^{(n)}(x) = \prod_{j \in \mathcal{C}_i} (x - \theta_n^j)   \]
be the minimal polynomial of $\theta_n^s$ over $\mathbb{F}_q$. 
Let $\mathcal{M}_s^{(n)}$ be the cyclic code of length $n$ over $\mathbb{F}_q$ generated by $(x^n - 1) / M_s^{(n)}(x)$. 
Here, $\mathcal{M}_s^{(n)}$ is called the $q$-ary irreducible cyclic code. 
Furthermore, if $\mathcal{C}_{s_1}, \mathcal{C}_{s_2}, \dots, \mathcal{C}_{s_k}$ are all the distinct $q$-cyclotomic cosets, then $\mathcal{M}_{s_1}^{(n)}, \mathcal{M}_{s_2}^{(n)}, \dots, \mathcal{M}_{s_k}^{(n)}$ are all the distinct $q$-ary irreducible cyclic codes of length $n$ over $\mathbb{F}_q$. 
We have the following result which is from \cite[Theorem 1]{FFTA-SB}. 
\begin{thm}
	Let $q$ be a prime power and $p$ be an odd prime co-prime to $q$. 
	Let $\theta_{n}$ be a primitive root modulo $n$ with $n = p^m$. 
	\begin{enumerate}[label=$\mathrm{(\arabic*)}$]
		\item The codes $\mathcal{M}_{0}^{(n)}$, $\mathcal{M}_{\theta_n^k p^j}^{(n)}$, $0 \leqslant j \leqslant m-1$, $0 \leqslant k \leqslant \frac{\varphi(p^{m-j})}{\ord_{p^{m-j}}(q)}$ are all the distinct $q$-ary irreducible cyclic codes of length $n = p^m$ over $\mathbb{F}_q$, where $\varphi$ is the Euler's Phi function. 
		
		\item The code $\mathcal{M}_{\theta_n^k p^j}^{(n)}$ is equivalent to the code $\mathcal{M}_{p^j}^{(n)}$ and therefore they have the same weight distribution. 
		
		\item All the non-zerocodewords in the code $\mathcal{M}_{0}^{(n)}$ have weight $p^m$. 
		
		\item The code $\mathcal{M}_{p^j}^{(n)}$ is the repetition code of the irreducible cyclic code $\mathcal{M}_{1}^{(p^{m-j})}$ of length $p^{m-j}$ coresponding to the $q$-cyclotomic coset $\mathcal{C}_1$, repeated $p^j$ times. 
		Furthermore, for each $w \geqslant 0$, we have 
		\[ A_{w}^{(p^m)} = 
		\begin{cases}
			0, & p^j \nmid w, \\
			A_{w'}^{(p^{m-j})}, & w = p^j w', 0 \leqslant w' \leqslant p^{m-j}, 
		\end{cases}
		\]
	\end{enumerate}
\end{thm}

\begin{rem}
	By \cite[Lemma 2]{FFTA-SBR}, the above theorem also holds when $p=2$. 
	Moreover, we have $A_{p^m}^{(p^m)} \neq 0$ if and only if $A_{p^{m-j}}^{(p^{m-j})} \neq 0$. 
	So we only need to consider the weight distribution of the code $\mathcal{M}_{1}^{(p^{r})}$ for some $r = m-j$. 
\end{rem}

In this section, we will show 

\begin{thm}
	Let $q$ be a prime power and $n = p^m$ be a prime power co-prime to $q$. 
	Then we have the following statements. 
	\begin{enumerate}[label=$\mathrm{(\arabic*)}$]
		\item 
		If $q \equiv 1 \pmod{4}$, then $h_q(2^m) = 0$. 
		
		\item 
		Let $q = -1 + 2^b c$, $b \geqslant 2$, $c$ odd, i.e., $q \equiv 3 \pmod{4}$. 
		If $r \geqslant 3$ and $t = \min\{ r, b+1 \}$, then $h_q(2^m) = 0$ when $q - 2^{t-1} + 1 \neq 0$ and $h_q(2^m) \geqslant 1$ when $q - 2^{t-1} + 1 = 0$. 
		
		\item 
		If $\ord_{p^m}(q) = \varphi(p^m)$, then $h_q(p^m) = 0$. 
		
		\item  
		If $\ord_{p^m}(q) = p^d$, then $h_q(p^m) = 0$. 
		
		\item  
		If $\ord_{p^m}(q) = 2p^d$ and $q - p^u + 1 \neq 0$, then $h_q(p^m) = 0$ where $u = \min\{ r, m-d \}$.
	\end{enumerate}
\end{thm}
\begin{proof}
%	\Cref{prop-q-2-1,prop-q-2-3,prop-q-p-1,prop-q-p-2,prop-q-p-3}
	See Propositions \ref{prop-q-2-1}--\ref{prop-q-p-3}. 
\end{proof}

\subsection{\texorpdfstring{$q \equiv 1 \pmod{4}$}{} and \texorpdfstring{$p = 2$}{}}

\begin{thm}[\texorpdfstring{\cite[Theorem 1]{FFTA-SBR}}{}]
	Let $q = 1 + 2^b c$, $b \geqslant 2$, $c$ odd, i.e., $q \equiv 1 \pmod{4}$. 
	\begin{enumerate}[label=$\mathrm{(\arabic*)}$]
		\item If $r \leqslant b$, then the only possible non-zero weight in $\mathcal{M}_{1}^{(2^r)}$ is $2^r$, which is attained by all its $q-1$ non-zero codewords. 
		
		\item If $r > b$, then the weight distribution $A_{w}^{(p^r)}$ of $\mathcal{M}_{1}^{(2^r)}$ is given by 
		\[ A_{w}^{(2^r)} = 
		\begin{cases}
			0, & 2^{b} \nmid w, \\
			\binom{2^{r-b}}{w'} (q-1)^{w'}, & w = 2^{b} w', 0 \leqslant w' \leqslant 2^{r-b}. 
		\end{cases}
		\]
	\end{enumerate}
\end{thm}

\begin{prop}\label{prop-q-2-1}
	If $q \equiv 1 \pmod{4}$, then $h_q(2^m) = 0$. 
\end{prop}

\subsection{\texorpdfstring{$q \equiv 3 \pmod{4}$}{} and \texorpdfstring{$p = 2$}{}}

\begin{thm}[\texorpdfstring{\cite[Theorem 2 \& Theorem 3]{FFTA-SBR}}{}]
	Let $q = -1 + 2^b c$, $b \geqslant 2$, $c$ odd, i.e., $q \equiv 3 \pmod{4}$. 
	\begin{enumerate}[label=$\mathrm{(\arabic*)}$]
		\item 
		If $r = 1$, then the weight distribution $A_{w}^{(2)}$ of $\mathcal{M}_{1}^{(2)}$ is given by 
		\[ A_{0}^{(2)} = 1, \quad A_{1}^{(2)} = 0, \quad A_{2}^{(2)} = q-1. \]
		
		\item 
		If $r = 2$, then the weight distribution $A_{w}^{(4)}$ of $\mathcal{M}_{1}^{(4)}$ is given by 
		\[ A_{0}^{(4)} = 1, \quad A_{1}^{(4)} = 0, \quad A_{2}^{(4)} = 2(q-1), \quad A_{3}^{(4)} = 0, \quad A_{4}^{(4)} = (q-1)^2. \]
		
		\item 
		If $r \geqslant 3$, then set $t = \min\{ r, b+1 \}$ and the weight distribution $A_{w}^{(2^r)}$ of $\mathcal{M}_{1}^{(2^r)}$ is given by 
		\[ A_{w}^{(2^r)} = \sum m(w_1) m(w_2) \dots m(w_{2^{r-t}}), \]
		where the summation runs over all tuples $(w_1, \dots, w_{2^{r-t}})$ of non-negative integers $w_i$'s such that $w_1 + w_2 + \dots + w_{2^{r-t}} = w$ and 
		\[
		m(v) = 
		\begin{cases}
			1, & v = 0, \\
			(q-1) 2^{t-1}, & v = 2^t - 2, \\
			(q-1)(q - 2^{t-1} + 1), & v = 2^t, \\
			0, & \mathrm{otherwise}.
		\end{cases}
		\]
	\end{enumerate}
\end{thm}

\begin{prop}\label{prop-q-2-3}
	Let $q = -1 + 2^b c$, $b \geqslant 2$, $c$ odd, i.e., $q \equiv 3 \pmod{4}$. 
	If $r \geqslant 3$ and $t = \min\{ r, b+1 \}$, then $h_q(2^m) = 0$ when $q - 2^{t-1} + 1 \neq 0$ and $h_q(2^m) \geqslant 1$ when $q - 2^{t-1} + 1 = 0$. 
\end{prop}
\begin{proof}
	We only need to determine $A_{2^r}^{(2^r)}$. 
	If $r \leqslant 2$, then $A_{2^r}^{(2^r)} > 0$.
	So we only need to consider the case $r > 2$. 
	
	If $t = r \leqslant b+1$, then $2^{r-t} = 1$. 
	Hence, 
	\[ A_{2^r}^{(2^r)} = (q-1)(q - 2^{t-1} + 1). \]
	
	If $t < r$, then there exists integers $a,b$ such that 
	\[ 2^{r} = a \cdot 2^t + b \cdot (2^t - 2). \]
	It is clear that $(a,b) = (2^{r-t}, 0)$ is one solution. 
	We claim that it is the unique solution such that $a,b \geqslant 0$ and $a + b \leqslant 2^{r-t}$.  
	Let $S = a+b \leqslant 2^{r-t}$. 
	Since $b = 2^{t-1} S - 2^{r-1} \geqslant 0$, then 
	\[ 2^{r-t} \leqslant S \leqslant 2^{r-t}. \]
	That is to say, $S$ is unique, which implies that $b$ is unique. 
	Hence, 
	\[ A_{p^r}^{(p^r)} = (n(2^t))^a = (q-1)^a (q - 2^{t-1} + 1)^a, \]
	which completes the proof. 
\end{proof}

\begin{coro}
	If $q = 3$, then $h_q(2^m) = 0$ if $m \leqslant 2$ and $h_q(2^m) > 0$ if $m \geqslant 3$. 
	Specially, $h_3(8) = 1$. 
\end{coro}

\subsection{\texorpdfstring{$\ord_{p^m}(q) = \varphi(p^m)$}{} with odd prime \texorpdfstring{$p$}{}}

To determine the weight distribution $A_{w}^{(p^r)}$ of $\mathcal{M}_{1}^{(p^r)}$, we need the following notations. 
For each $t \geqslant 1$, $v \geqslant 2$, let 
\[ P_t(v) = \left\{ \vec{v} = (v_1, \dots, v_t) \in \mathbb{Z}^t \ | \ 2 \leqslant v_j \leqslant p, \sum_{j=1}^{t} v_j = v \right\}. \]
For given $\vec{v} = (v_1, \dots, v_t) \in P_t(v)$, set 
\[ L(\vec{v}) = \left\{ \vec{\ell} = (\ell_1, \dots, \ell_t) \in \mathbb{Z}^t \ | \ \ell_j \geqslant v_j - 2, \sum_{j=1}^{t} \ell_j \leqslant p-2t \right\}. \]
For given $\vec{\ell} = (\ell_1, \dots, \ell_t) \in L(\vec{v})$, put $A(\vec{v}; \vec{\ell})$ to be equal to 
\[ a(\vec{\ell}) \binom{\ell_1}{v_1 - 2} \binom{\ell_2}{v_2 - 2} \dots \binom{\ell_t}{v_t - 2} (q-1)^t (q-2)^{v-2t}, \]
where
\[ a(\vec{v}) = \sum_{k_1 = 1}^{p-2t+1-\sum_{i=1}^{t}\ell_i} \sum_{k_2 = k_1 + \ell_1 + 2}^{p-2(t-1)+1-\sum_{i=2}^{t}\ell_i} \dots \sum_{k_{t-1} = k_{t-2}+\ell_{t-2}+2}^{p-3-\sum_{i=t-1}^{t}\ell_i} \sum_{k_{t} = k_{t-1}+\ell_{t-1}+2}^{p-1-\ell_t} 1. \]

\begin{thm}[\texorpdfstring{\cite[Theorem 2]{FFTA-SB}}{}] 
	Let $q$ be a prime power and $p$ be an odd prime co-prime to $q$. 
	If $\ord_{p^m}(q) = \varphi(p^m)$, then the weight distribution $A_{w}^{(p^r)}$ of $\mathcal{M}_{1}^{(p^r)}$ is given by 
	\[ A_{w}^{(p^r)} = \sum N(w_1) N(w_2) \cdots N(w_{p^{r-1}}), \]
	where the summation runs over all tuples $(w_1, \dots, w_{p^{r-1}})$ of non-negative integers $w_i$'s such that $w_1 + w_2 + \dots + w_{p^{r-1}} = w$ and 
	\[
	N(v) = 
	\begin{cases}
		1, & v = 0, \\
		0, & v = 1~\mathrm{or}~v \geqslant p+1, \\
		\sum_{t \geqslant 1} \sum_{\vec{v} \in P_t(v)} \sum_{\vec{\ell} \in L(\vec{v})} A(\vec{v}; \vec{\ell}), & \mathrm{otherwise}. 
	\end{cases}
	\]
\end{thm}

\begin{prop}\label{prop-q-p-1}
	Let $q$ be a prime power and $p$ be an odd prime co-prime to $q$. 
	If $\ord_{p^m}(q) = \varphi(p^m)$, then $h_q(p^m) = 0$. 
\end{prop}
\begin{proof}
	We only need to determine $A_{p^r}^{(p^r)}$. 
	Let $w_1 = w_2 = \dots = w_{p^{r-1}} = w$ be the solution of $\sum_{i=1}^{p^{r-1}} = p^r$. 
	Since $N(\cdot)$ is non-negative, then $A_{p^r}^{(p^r)} \geqslant N(p)^{p^{r-1}}$. 
	Let $p = 3(t-1) + p'$ for some integer $t$ with $p' \in \{ 3, 4, 5 \}$. 
	It is clear that 
	\[ \vec{v} = (3,3,\dots,3,p') \in P_t(p), \quad \vec{\ell} = (1,1,\dots, p'-2) \in L(\vec{v}). \]
	Since 
	\[ p - \sum_{i=1}^{t} \ell_i - 2t + 1 = p - (p'+t-3) - 2t + 1 = 1 > 0 \]
	then $a_{\vec{\ell}} > 0$, which implies that 
	\[ A_{p^r}^{(p^r)} \geqslant N(p) \geqslant A(\vec{v}; \vec{\ell}) > 0, \]
	which completes the proof. 
\end{proof}

\begin{rem}
	In fact, \Cref{prop-q-p-1} has also been proved in \cite[Theorem 3.3]{arxiv-LYLZ} by a different method. 
\end{rem}

\subsection{\texorpdfstring{$\ord_{p^m}(q) = p^d$}{} with odd prime \texorpdfstring{$p$}{}}

\begin{thm}[\texorpdfstring{\cite[Theorem 3]{FFTA-SB}}{}] 
	Let $q$ be a prime power and $p$ be an odd prime co-prime to $q$. 
	Assume that $\ord_{p^m}(q) = p^d$ for some integer $d$ with $0 \leqslant d \leqslant m$. 
	\begin{enumerate}[label=$\mathrm{(\arabic*)}$]
		\item If $r \leqslant m - d$, then the only possible non-zero weight in $\mathcal{M}_{1}^{(p^r)}$ is $p^r$, which is attained by all its $q-1$ non-zero codewords. 
		
		\item If $r > m - d$, then the weight distribution $A_{w}^{(p^r)}$ of $\mathcal{M}_{1}^{(p^r)}$ is given by 
		\[ A_{w}^{(p^r)} = 
		\begin{cases}
			0, & p^{m-d} \nmid w, \\
			\binom{p^{r-(m-d)}}{w'} (q-1)^{w'}, & w = p^{m-d} w', 0 \leqslant w' \leqslant p^{r-(m-d)}. 
		\end{cases}
		\]
	\end{enumerate}
\end{thm}

\begin{prop}\label{prop-q-p-2}
	If $\ord_{p^m}(q) = p^d$, then $h_q(p^m) = 0$. 
\end{prop}
\begin{proof}
	We only need to determine $A_{p^r}^{(p^r)}$. 
	It is clear that $w = p^r = p^{m-d} w'$ with non-zero $w' = p^{r-(m-d)}$. 
	Then 
	\[ A_{p^r}^{(p^r)} = \binom{p^{r-(m-d)}}{w'} (q-1)^{w'} = (q-1)^{p^{r-(m-d)}} > 0,  \]
	which completes the proof. 
\end{proof}

\subsection{\texorpdfstring{$\ord_{p^m}(q) = 2p^d$}{} with odd prime \texorpdfstring{$p$}{}}

\begin{thm}[\texorpdfstring{\cite[Theorem 4]{FFTA-SB}}{}] 
	Let $q$ be a prime power and $p$ be an odd prime co-prime to $q$. 
	If $\ord_{p^m}(q) = 2 p^d$ for some integer $d$ with $0 \leqslant d \leqslant m$, then the weight distribution $A_{w}^{(p^r)}$ of $\mathcal{M}_{1}^{(p^r)}$ is given by 
	\[ A_{w}^{(p^r)} = \sum n(w_1) n(w_2) \dots n(w_{p^{r-u}}), \]
	where 
	\[
	n(v) = 
	\begin{cases}
		1, & v = 0, \\
		(q-1) p^u, & v = p^u - 1, \\
		(q-1)(q - p^u + 1), & v = p^u, \\
		0, & \mathrm{otherwise}, 
	\end{cases}
	\]
	with $u = \min\{ r, m-d \}$. 
\end{thm}

\begin{prop}\label{prop-q-p-3}
	If $\ord_{p^m}(q) = 2p^d$, then $h_q(p^m) = 0$ when $q - p^u + 1 \neq 0$ and $h_q(p^m) \geqslant 1$ when $q - p^u + 1 = 0$. 
\end{prop}
\begin{proof}
	We only need to determine $A_{p^r}^{(p^r)}$. 
	If $u = r$, then 
	\[ A_{p^r}^{(p^r)} = n(p^r) = n(p^u) = (q-1)(q - p^u + 1). \]
	If $u = m - d > r$, then there exists integers $a,b$ such that 
	\[ p^r = a \cdot p^u + b \cdot (p^u - 1). \]
	It is clear that $(a,b) = (p^{r-u}, 0)$ is one solution. 
	
	We claim that it is the unique solution such that $a,b \geqslant 0$ and $a + b \leqslant p^{r-u}$.  
	
	Let $S = a+b \leqslant p^{r-u}$. 
	Since $b = S p^u - p^r \geqslant 0$, then 
	\[ p^{r-u} \leqslant S \leqslant p^{r-u}. \]
	That is to say, $S$ is unique, which implies that $b = S p^u - p^r$ is unique. 
	Hence, 
	\[ A_{p^r}^{(p^r)} = (n(p^u))^a = n(p^u) = (q-1)^a (q - p^u + 1)^a, \]
	which completes the proof. 
\end{proof}

\section{Conclusion}\label{sec-conclusion}
In this paper, we have determined the values on $h_q(n)$ by irreducible cyclic codes, constacyclic codes, support weight distribution in coding theory. 
We found that $h_q(n) = 0$ if and only every non-trivial cyclic code contains a codeword of Hamming weight $n$, which has been discovered by different methods in the literature. 
By new methods, we proved $h_2(2p) = 2$ if $2$ is a primitive root modulo the prime $p$. 
Moreover, the support weight distribution and constacyclic codes also make sense when determining $h_q(n)$. 
Finally, using irreducible cyclic codes, we have derived several families of $n$ such that $h_q(n) = 0$, which covers some results in the literature and develops more on trivial cyclically covering subspaces.

\section*{Declarations}  The authors declare no conflict of interest.


\begin{thebibliography}{00}
	\bibitem{JCTA-AGJ} J. Aaronson, C. Groenland, and T. Johnston. Cyclically covering subspaces in $\mathbb{F}_2^n$. J. Comb. Theory Ser. A 181 (2021): 105436.
	
	\bibitem{EJC-CER} P. J. Cameron, D. Ellis, and W. Raynaud. Smallest cyclically covering subspaces of $\mathbb{F}_q^n$, and lower bounds in Isbell's conjecture. Eur. J. Comb. 81 (2019): 242--255.
	
	\bibitem{PBC-DM-C} P. J. Cameron. Problems from the Thirteenth British Combinatorial Conference. Discrete Math. 125 (1994): 407--417.
	
	\bibitem{FFTA-CFLL} B. Chen, Y. Fan, L. Lin, and H. Liu. Constacyclic codes over finite fields. Finite Fields and Their Applications 18 (2012): 1217--1231.
	
	\bibitem{TIT-D} C. Ding. The weight distribution of some irreducible cyclic codes. IEEE Trans. Inf. Theory 55(3) (2009): 955--960.
	
	\bibitem{SAM-G} C. Greene. Weight enumeration and the geometry of linear codes. Stud. Appl. Math. 55(2) (1976): 119--128.
	
	\bibitem{DM-HKM} T. Helleseth, T. Kløve, and J. Mykkeltveit. The weight distribution of irreducible cyclic codes with block lengths $n_1 ((q^{l}-1)/N)$. Discrete Math. 18(2) (1977): 179--211.
	
	\bibitem{TIT-HKY} T. Helleseth, T. Kløve, and O. Ytrehus. Generalized Hamming weights of linear codes. IEEE Trans. Inf. Theory 38(3) (1992): 1133--1140.
	
	\bibitem{JRAM-H} C. Hooley. On Artin's conjecture. J. Reine Angew. Math. 225 (1967): 209--220.
	
	\bibitem{FFTA-H} J. Huang. On trivial cyclically covering subspaces of $\mathbb{F}_q^n$. Finite Fields and Their Applications 96 (2024): 102423.
	
	\bibitem{Huffman-Pless} W. C. Huffman and V. Pless. Fundamentals Error Correcting Codes. Cambridge University Press, 2003.
	
	\bibitem{MS-Isbell-1} J. R. Isbell. Homogeneous games. Math. Stud. 25 (1957): 123--128.
	
	\bibitem{PAMS-Isbell-2} J. R. Isbell. Homogeneous games, II. Proc. Am. Math. Soc. 11 (1960): 159--161.
	
	\bibitem{JCTA-J} R. E. Jamison. Covering finite fields with cosets of subspaces. J. Comb. Theory Ser. A 22 (1977): 253--266.
	
	\bibitem{DM-Klove} T. Kløve. The weight distribution of linear codes over GF($q^l$) having generator matrix over GF$(q)^{\ast}$. Discrete Math. 23(2) (1978): 159--168.
	
	\bibitem{arxiv-LY-1} S. Li and P. Yuan. On trivial cyclically covering subspaces of $\mathbb{F}_q^n$ in non-coprime characteristic. arXiv:2512.24301.
	
	\bibitem{arxiv-LY-2} Y. Li and P. Yuan. Discrete Fourier transform approach to cyclically covering subspaces of $\mathbb{F}_q^n$. arXiv:2606.13307.
	
	\bibitem{arxiv-LY-3} Y. Li and P. Yuan. Cyclic codes and cyclically covering subspaces over finite fields. arXiv:2607.02239.
	
	\bibitem{arxiv-LYLZ} Y. Li, P. Yuan, S. Li, and Y. Zeng. On cyclically covering subspaces of $\mathbb{F}_q^n$. arXiv:2602.04558.
	
	\bibitem{TIT-L} J. H. van Lint. Repeated-root cyclic codes. IEEE Trans. Inf. Theory 37(2) (1991): 343--345.
	
	\bibitem{Book-Finite_Fields} R. Lidl and H. Niederreiter. Finite Fields. Cambridge University Press, 1997.
	
	\bibitem{AMASH-Luh} J. Luh. On the representation of vector spaces as a finite union of subspaces. Acta Math. Acad. Sci. Hungar. 23 (1972): 341--342.
	
	\bibitem{FFTA-SB} A. Sharma and G. K. Bakshi. The weight distribution of some irreducible cyclic codes. Finite Fields and Their Applications 18(1) (2012): 144--159.
	
	\bibitem{FFTA-SBR} A. Sharma, G. K. Bakshi, and M. Raka. The weight distribution of irreducible cyclic codes of length $2^m$. Finite Fields and Their Applications 13(1) (2007): 1086--1095.
	
	\bibitem{TIT-SLH} M. Shi, S. Li, and T. Helleseth. The weight enumerator polynomials of the lifted codes of the projective Solomon-Stiffler codes. IEEE Trans. Inf. Theory 70(9) (2024): 343--345.
	
	\bibitem{FFTA-SMZ} M. Sun, C. Ma, and L. Zeng. Cyclically covering subspaces of $\mathbb{F}_q^n$. Finite Fields and Their Applications 106 (2025): 102625.
	
	\bibitem{TIT-Wei} V. K. Wei. Generalized Hamming weights for linear codes. IEEE Trans. Inf. Theory 37(5) (1991): 1412--1418.
	
%	\bibitem{JCTA-AGJ} J. Aaronson, C. Groenland, T. Johnston, Cyclically covering subspaces in $\mathbb{F}_2^n$, J. Comb. Theory, Ser. A 181 (2021) 105436. 
%	
%	\bibitem{EJC-CER} P.J. Cameron, D. Ellis, W. Raynaud, Smallest cyclically covering subspaces of $\mathbb{F}_q^n$, and lower bounds in Isbell’s conjecture, Eur. J. Comb. 81 (2019) 242--255. 
%	
%	\bibitem{PBC-DM-C} P.J. Cameron, Problems from the Thirteenth British Combinatorial Conference, Discrete Math., 125: 407–417 (1994).
%	
%	\bibitem{FFTA-CFLL} Bocong Chen, Yun Fan, Liren Lin, Hongwei Liu, Constacyclic codes over finite fields, Finite Fields and Their Applications, 18: 1217-1231 (2012). 
%	
%	\bibitem{TIT-D} Cunsheng Ding, The weight distribution of some irreducible cyclic codes, IEEE Trans. Inf. Theory, 55(3): 955-960 (2009). 
%	
%	\bibitem{SAM-G} C. Greene, Weight enumeration and the geometry of linear codes, Stud. Appl. Math., vol. 55, no. 2, pp. 119–128, Jun. 1976.
%	
%	\bibitem{DM-HKM} T. Helleseth, T. Kløve, J. Mykkeltveit, The weight distribution of irreducible cyclic codes with block lengths $n_1 ((q^{l}-1)/N)$, Discrete Mathematics, vol. 18, no. 2, pp. 179–211, 1977.
%	
%	\bibitem{TIT-HKY} T. Helleseth, T. Klove, and O. Ytrehus, Generalized Hamming weights of linear codes, IEEE Trans. Inf. Theory, vol. 38, no. 3, pp. 1133–1140, May 1992.
%	
%	\bibitem{JRAM-H} C. Hooley, On Artin’s conjecture, J. Reine Angew. Math. 225: 209–220 (1967).
%	
%	\bibitem{FFTA-H} J. Huang, On trivial cyclically covering subspaces of $\mathbb{F}_q^n$, Finite Fields and Their Applications, 96: 102423 (2024). 
%	
%	\bibitem{Huffman-Pless} W. C. Huffman, V. Pless, Fundamentals Error Correcting Codes, Cambridge, U.K.: Cambridge Univ. Press, 2003. 
%	
%	\bibitem{MS-Isbell-1} J.R. Isbell, Homogeneous games, Math. Stud. 25: 123-128 (1957). 
%	\bibitem{PAMS-Isbell-2} J.R. Isbell, Homogeneous games, II, Proc. Am. Math. Soc. 11: 159-161 (1960).
%	
%	\bibitem{JCTA-J} R.E. Jamison, Covering finite fields with cosets of subspaces, J. Comb. Theory A 22: 253–266 (1977).
%	
%	\bibitem{DM-Klove} Kløve, The weight distribution of linear codes over GF($q^l$) having generator matrix over GF$(q)^{\ast}$, Discrete Math., vol. 23, no. 2, pp. 159–168, 1978.
%	
%	
%	\bibitem{arxiv-LY-1} S. Li, P. Yuan, On trivial cyclically covering subspaces of $\mathbb{F}_q^n$ in non-coprime characteristic, arXiv:2512.24301. 
%	
%	\bibitem{arxiv-LY-2} Yangcheng Li, Pingzhi Yuan, Discrete Fourier Transform Approach to Cyclically Covering Subspaces of $\mathbb{F}_q^n$, arXiv:2606.13307. 
%	
%	\bibitem{arxiv-LYLZ} Yangcheng Li, Pingzhi Yuan, Shuang Li, Yuanpeng Zeng, On cyclically covering subspaces of $\mathbb{F}_q^n$, arXiv:2602.04558. 
%	
%	\bibitem{TIT-L} J. H. van Lint, Repeated-root cyclic codes, IEEE Trans. Inf. Theory, 37(2): 343-345 (1991). 
%	 
%	\bibitem{Book-Finite_Fields} R. Lidl, H. Niederreiter: Finite Fields. Cambridge University Press, Cambridge (1997). 
%	
%	\bibitem{AMASH-Luh} J. Luh, On the representation of vector spaces as a finite union of subspaces, Acta Math. Acad. Sci. Hungar. 23: 341–342 (1972).
%	
%	\bibitem{FFTA-SB} Anuradha Sharma, Gurmeet K. Bakshi, The weight distribution of some irreducible cyclic codes, Finite Fields and Their Applications, 18(1): 144-159 (2012). 
%	
%	\bibitem{FFTA-SBR} Anuradha Sharma, Gurmeet K. Bakshi, Madhu Raka, The weight distribution of irreducible cyclic codes of length $2^m$, Finite Fields and Their Applications, 13(1): 1086-1095 (2007). 
%	
%    \bibitem{TIT-SLH} Minjia Shi, Shitao Li, Tor Helleseth, The weight enumerator polynomials of the lifted
%    codes of the projective Solomon-Stiffler codes, IEEE Trans. Inf. Theory, 70(9): 343-345 (2024). 
%    
%    \bibitem{FFTA-SMZ} Meng Sun, Changli Ma, Liwei Zeng, Cyclically covering subspaces of $\mathbb{F}_q^n$, Finite Fields and Their Applications, 106: 102625 (2025). 
%   
%   \bibitem{TIT-Wei} V.K. Wei, Generalized Hamming weights for linear codes, IEEE Trans. Inf. Theory, 37(5): 1412–1418 (1991). 
\end{thebibliography}
\end{document}